\documentclass[runningheads]{llncs}
\usepackage{graphicx}
\usepackage{amssymb}
\usepackage{amsmath}
\usepackage{amsfonts}
\usepackage{tikz}
\usepackage{pgfplots}
\pgfplotsset{width=10cm,compat=1.9}
\usepackage{hyperref}

\makeatletter
\renewcommand*{\@fnsymbol}[1]{\ensuremath{\ifcase#1\or *\or \dagger\or \ddagger\or
   \mathsection\or \mathparagraph\or \|\or **\or \dagger\dagger
   \or \ddagger\ddagger \else\@ctrerr\fi}}
\makeatother

\usepackage[capitalise]{cleveref}
\usepackage{mathtools} 
\usepackage{mathdots}

\definecolor{amber}{rgb}{1.0, 0.75, 0.0}

\begin{document}
\title{Optimal MEV Extraction Using \\Absolute Commitments}

\newcommand{\PSPACE}{\textsf{PSPACE}}

\author{Daji Landis${}^1$\thanks{Supported by Flashbots Research (FRP-41) and the Ethereum Foundation (ROP-7).} \and Nikolaj I. Schwartzbach${}^2$\thanks{Supported by the European Research Council (ERC) under the European Union’s Horizon 2020 research and innovation programme (Grant agreement No. 101019547).}${}^*$}
\authorrunning{Daji Landis and Nikolaj I. Schwartzbach}
\institute{New York University \and Bocconi University}

\let\oldaddcontentsline\addcontentsline
\def\addcontentsline#1#2#3{}
\maketitle
\def\addcontentsline#1#2#3{\oldaddcontentsline{#1}{#2}{#3}}

\begin{abstract}
We propose a new, more potent attack on decentralized exchanges. This attack leverages absolute commitments, which are commitments that can condition on the strategies made by other agents. This attack allows an adversary to charge monopoly prices by committing to undercut those other miners that refuse to charge an even higher fee. This allows the miner to extract the maximum possible price from the user, potentially through side channels that evade the inefficiencies and fees usually incurred. This is considerably more efficient than the prevailing strategy of `sandwich attacks', wherein the adversary induces and profits from fluctuations in the market price to the detriment of users. The attack we propose can, in principle, be realized by the irrevocable and self-executing nature of smart contracts, which are readily available on many major blockchains. Thus, the attack could potentially be used against a decentralized exchange and could drastically reduce the utility of the affected exchange. 

\end{abstract}

\section{Introduction}
Decentralized finance (DeFi) is a broad term for infrastructure for financial transactions built with the aim not to be dependent on the reliability or honesty of a single centralized organization \cite{zetzsche2020decentralized}. Fittingly, there is not just one central currency in the world of DeFi and the linchpin of the infrastructure is the capacity to exchange different currencies for one another, much like traditional currency exchange. There are centralized exchanges that are run by a single, centralized party, and decentralized exchanges (DEXs) that use interactive cryptographic protocols to allow for decentralized, permissionless exchange \cite{malamud2017decentralized}. Since DEXs supposedly do not have any oversight, their well-functioning is totally predicated on the soundness of their governing protocols \cite{patel2014pure}.
Decentralized exchanges (DEXs) allow for permissionless exchange of cryptocurrencies.  They are implemented on some underlying blockchain in the form of a smart contract and their  well-functioning is fully predicated on the soundness of their governing protocols \cite{patel2014pure}.
These protocols have been found to have weaknesses that allow certain agents to extract substantial value from the system \cite{flash_boys}.  These exploitations do not stem from a breach in security, but rather from flaws in the incentive structure of the protocol.  The protocol itself is not broken or disobeyed; instead, the incentives of the game lead to unintended consequences.  
When interacting with a DEX, a user proposes to exchange an amount of currency $X$ for an amount of currency $Y$. This offer is publicly announced, but is not completed immediately, and the exchange rate may change during the delay. To account for the variability in the exchange rate, the user also specifies a \emph{slippage} that sets a lower bound on the amount of $Y$ they are willing to receive. Unfortunately, this system is vulnerable to an attack wherein other users, known as \emph{flash bots}, will `sandwich' the users' transactions with their own \cite{flash_boys}. This enables the bots to change the exchange rates to their benefit, while pushing the user's transaction to the limit of their slippage allowance \cite{sandwich}. The value extracted by these flash bots is known as \emph{miner extractable value} (MEV) and it makes up a non-trivial fraction of the cost incurred by the users of DEXs.

Smart contracts are key to the implementation of DEXs, as well as other blockchain services. Once a smart contract is deployed, it cannot be altered, even by its creators. Thus, users can make commitment moves with smart contracts that lock them in to certain moves given stipulated circumstances. Furthermore, the state of all smart contracts is public, meaning that these commitments can condition on the smart contracts deployed by other agents. This allows for powerful \emph{absolute commitments} \cite{stack_attack_mathias,stack_attack_aamas}, which are known to enable powerful attacks against auctions and transaction fee mechanisms \cite{stack_attack_ecai}, as well as certain escrow mechanisms \cite{stack_attack_aamas}. In this paper, we build on this line of work and analyze a model of automated market makers (AMMs), a popular type of DEX. We ask the following natural question.

 \begin{quote}
    \em Do absolute commitments allow for more efficient miner extractable value in decentralized exchanges?
 \end{quote}

We answer this question in the affirmative (\cref{thm:main}), given the assumptions in our `popsicle game' model of AMMs. This model captures the interactions of miners, who we call vendors, and their customers.  Using our model, we analyze what happens when we extend absolute commitment capabilities to the agents. We find a powerful attack whereby the first vendor can charge monopoly prices by forcing all other vendors to commit to exorbitant prices.



In terms of practicality, we believe that the attack we propose is not feasible to deploy today, owing in part to a limited instruction set of most blockchain virtual machines. Rather, we show that agents are strongly incentivized to execute these attacks if they are at all possible. Indeed, it is our worry that future enhanced capabilities of blockchains may one day enable such attacks. It is our goal to draw attention to this problem while it still remains theoretical, in the hope that the community can find ways to mitigate the issue.

 \subsection{Related Work}

The incentives that underlie the well-functioning of blockchain processes -- in particular MEV and the subtle powers of miners and bots -- have been studied extensively.  The seminal investigation of these phenomenon by Daian et al. \cite{flash_boys} discusses the mechanisms that allow for these unfair market practices.  One of these practices is frontrunning, where proprietary knowledge of a large transaction allows a party to either buy an asset before the price goes up due to the transaction or sell it before the price goes down. Daian et al. model how the ability of miners to reorder blocks and insert their own transactions allows for complex, automatic frontrunning. However, they show that the mechanism of sandwich attacks does not result in perfect extraction by the miner.  While the profits of the miners are still substantial, some of the profit leaks into the AMM at large \cite{sandwich}. We take this inefficiency into account in our model and note that it can be circumvented using absolute commitments.

Daian et al. further discusses how arbitrage results in cooperative bidding. The authors suggest that this cooperation arises organically \cite{flash_boys}. There are other accounts, both theoretical and empirical, of instances in which unexpected cooperation might arise. The model and evidence presented in \cite{bitcoin_collusion} indicate that Bitcoin miners may be pursuing strategies other than competitive mining.  In a case with perfect competition, we would expect every block to be full, given that there are always more pending transactions than there are spots in a block.  This is not observed in practice, suggesting that strategic capacity management is being used by miners to keep fees artificially high \cite{bitcoin_collusion}.  Conversely, in \cite{jens2}, the authors show that there can be a reciprocity relationship between miners that are taking part in a mining pool. The legality and the capabilities of collusion using a blockchain are discussed in \cite{legal_take_on_collusion,legal_take_on_collusion_2}.  Together, these works show that cooperation opportunities on blockchains are significant and their existence has important implications. 
 
\section{Our Scenario: MEV and Popsicles}

In our model, we focus only on specific aspects of the complicated workings of a DEX and an MEV attack. Users submit their transaction information, including a slippage parameter, to the \emph{mempool} of all pending transactions. Miners are chosen by some process, and each miner in turn strategically fills their blocks with transactions (which can include their own), and stipulates the order in which these transactions are processed. The included transactions are then executed, as long as the slippage parameters are not violated.  At first glance, this set up seems like an sealed-bid, first price auction where the slippage is the bid, but note that the amount the miner stands to gain is not the precise slippage amount. The value of MEV to the miner varies with the size of the transaction and the slippage amount \cite{sandwich}; the most profitable transactions might not be those with the highest slippage tolerance.  Thus, we must carefully think through our model, which focuses specifically on profit incentives.
First, we note that blocks are perfect substitutes -- either the transaction is executed or not.  Only two things vary across blocks: the time and the price. The provider of the good, the block miner, cannot control the timing of their block, this is handed down to them by the protocol \cite{eth_book}.  They can, however, somewhat control their price by refusing to execute any transaction that does not meet some slippage threshold. Additionally, we note that MEV is inefficient and cannot be levied against all users in a block as it requires two transaction slots, along with accompanying fees \cite{sandwich}.  Moving forward, we parameterize this inefficiency using the parameter $0 \leq \alpha <1$, where $\alpha=0$ means that no fees or other friction is felt by the miner, and note that a direct payment from the user to the miner would circumvent these inefficiencies.  We hold open the option of such side payments, but note this will not be relevant in our `vanilla' model.

In practice, the identities of the block miners are known in advance for some constant number of blocks. Ethereum publicly selects its validators two epochs (each consisting of 32 slots) ahead of time \cite{eth_book}.  If we let $n$ be the number of known, upcoming miners, the miners and their blocks can be labeled in order $1, \ldots , n$, where $1$ is the next block. We specify that block $t$ comes at time $t$ and charges price $p_t$. We assume that the buyer has some discount factor $d \in [0,1]$, where $d=1$ indicates indifference to delay.  We define the utility of a completed, free, and timely transaction as 1. The utility experience by the buyer is therefore one minus the price, modified by their discount factor, i.e. $(1-p_t)\cdot d^{t-1}$. 

As described, our model bears resemblance to the Hotelling's linear city model \cite{hotelling1929stability}, except that firms compete across time rather than distance. The canonical example of Hotelling's model is mobile vendors (selling, say, popsicles) along a beach.  Our model is different in that our vendors are not mobile; their block time is predetermined.  We can, however, take inspiration from this motivating example and imagine our miners as vendors along a one-way boardwalk
, see \cref{fig:boardwalk} for an illustration
. The customers are uniformly dispersed along this boardwalk, which partitioned into discrete steps of unit length.  Each customer is hungry for a popsicle and each has the same discount rate $d$. If a customer, with discount rate $d<1$, is standing in front of vendor $1$, a free popsicle there will result in utility $1$, whereas a free popsicle at the next vendor will result in utility $d$. The next will be even worse with $d^2$. If vendor $1$ sets their price to $p_1 = 1-d-\varepsilon$, then they will always beat the subsequent vendor, even if that vendor offers free popsicles. 


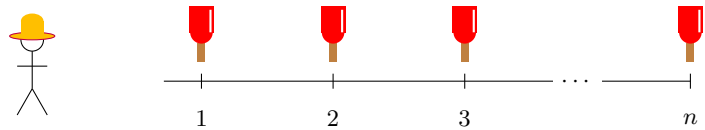
\begin{figure}
    \centering
    \begin{tikzpicture}
    \draw (0,0) -- (7,0);
    
    \begin{scope}[yshift=-0.1cm,xshift=-0.75cm]

        \draw (-1,0) -- (-1.2,-0.35);
        \draw (-1,0) -- (-0.8,-0.35);
        \draw (-1,0) -- (-1,0.5);
        \draw (-0.8,0.3) -- (-1.2,0.3);
        \draw (-1,0.65) circle [radius=0.15];
        
        \draw[fill=amber, draw=purple] (-1, 0.7) ellipse (0.3 and 0.05);
        \draw[fill=amber, draw=none] (-1, 0.9) ellipse (0.15 and 0.1);
        \draw[fill=amber, draw=none] (-1.15, 0.7) rectangle (-0.85, 0.9);

    \end{scope}
    
    \begin{scope}[xshift=0.5cm, yshift=0.5cm]
        \draw[fill=brown, draw=none ] (-0.05,-0.25) rectangle (0.05,0);
        
        \draw[fill=red, draw=none] (-0.16,0.15) rectangle (0.16,0.5);
        \draw[fill=red, draw=none] (-0.15,0.15) arc[start angle=180, end angle=360, radius=0.15];

        \draw[white, line width=0.8pt] (0.12,0.15) -- (0.12,0.45);

        \node at (0,-1) {$1$};
        \draw (0,-0.6) -- (0,-0.4);
    \end{scope}
    
    \begin{scope}[xshift=2.25cm, yshift=0.5cm]
        \draw[fill=brown, draw=none ] (-0.05,-0.25) rectangle (0.05,0);
        
        \draw[fill=red, draw=none] (-0.16,0.15) rectangle (0.16,0.5);
        \draw[fill=red, draw=none] (-0.15,0.15) arc[start angle=180, end angle=360, radius=0.15];

        \draw[white, line width=0.8pt] (0.12,0.15) -- (0.12,0.45);

        \node at (0,-1) {$2$};
        \draw (0,-0.6) -- (0,-0.4);
    \end{scope}
    
    \begin{scope}[xshift=4cm, yshift=0.5cm]
        \draw[fill=brown, draw=none ] (-0.05,-0.25) rectangle (0.05,0);
        
        \draw[fill=red, draw=none] (-0.16,0.15) rectangle (0.16,0.5);
        \draw[fill=red, draw=none] (-0.15,0.15) arc[start angle=180, end angle=360, radius=0.15];

        \draw[white, line width=0.8pt] (0.12,0.15) -- (0.12,0.45);

        \node at (0,-1) {$3$};
        \draw (0,-0.6) -- (0,-0.4);
    \end{scope}
    
    \begin{scope}[xshift=7cm, yshift=0.5cm]
        \draw[fill=brown, draw=none ] (-0.05,-0.25) rectangle (0.05,0);
        
        \draw[fill=red, draw=none] (-0.16,0.15) rectangle (0.16,0.5);
        \draw[fill=red, draw=none] (-0.15,0.15) arc[start angle=180, end angle=360, radius=0.15];

        \draw[white, line width=0.8pt] (0.12,0.15) -- (0.12,0.45);

        \node at (0,-1) {$n$};
        \draw (0,-0.6) -- (0,-0.4);
    \end{scope}

    \node[fill=white] at (5.5,0) {$\cdots$};
    \end{tikzpicture}
    \caption{Schematic illustration of the boardwalk. The popsicles  represent the blocks in which the buyer can have their transaction included, and $p_i$ is the price to the buyer for having their transaction included in the $i^{th}$ block, or in this case, the $i^{th}$ popsicle. Length along the boardwalk represents the passage of time.}
    \label{fig:boardwalk}
\end{figure}

We should take a moment to be careful with our analogy; blockchain users cannot specify which block executes their transaction. Rather, they can specify a slippage tolerance, which acts as a ceiling for the extraction they are willing to abide. If vendor $3$ were giving out free popsicles while all others charged some price $p>\varepsilon>0$, then a buyer starting at $t=1$ could not specify they wanted to be served by vendor $3$ directly.  However, they could set their slippage to $s=\varepsilon$.  Then all the vendors selling at $p$ would pass the buyer on until they arrived at vendor $3$, where they would be served.  If vendor $4$ offered the same deal, our buyer starting at $t=1$ would never reach vendor $4$, regardless of their discount rate, as they would be served first by vendor $3$. Thus, specifying slippage amounts to picking a vendor as long as that vendor has a price less than that of the preceding vendors. In practice, although miners cannot directly declare their floor prices, there are ways that the miners can signal their expectations. In particular, users will know if their transaction is never selected by a miner because they do not meet the floor price. Services, including the automatic slippage set in Uniswap, aim to keep users apprised of current market rates. As discussed in \cite{bitcoin_collusion}, capacity management can set those market rates.  We will, however, assume users can pick miners, who can broadcast their prices, as this is a straight forward model that still captures the important strategic aspects.



\subsection{Formal Model}
We now give a formal model of the popsicle game. Consider a single buyer that wants to obtain a popsicle that imparts a utility of 1. We consider a set of $n \geq 2$ vendors that we identify using the set $[n]=\{1,2,\ldots,n\}$, and assign the buyer index 0. The game starts with each vendor $i$ simultaneously deciding on a popsicle price $p_i \in [0,1]$. We assume the marginal cost for producing a popsicle is zero. The popsicles are assumed to be perfect substitutes; the buyer would be equally happy with any popsicles. We model the game as an extensive-form game of imperfect information. Here, we model games as trees, the branches of which are labeled with the index of a player. The leaves correspond to outcomes and are labeled with a utility for each player. The game is played starting at the root. At each node, the corresponding player chooses a subtree (henceforth called a \emph{subgame}) to recurse into. The procedure stops when we reach a leaf and distribute the corresponding utilities to the players. 

The \emph{popsicle game} starts with a move for vendor 1, whose action space is $[0,1]$, deciding their price $p_1$. This is followed by a move for vendor 2, who likewise chooses their price. To disallow vendor 2 conditioning their price on the price of vendor 1, all moves for vendor 2 are in the same information set. The game proceeds similarly for the remaining vendors. After the vendors have set their prices, the buyer learns the prices of all the vendors and selects a vendor. Their action space is the set $[n]$. We imagine the buyer located at $x=0$ on the real number line, with vendor $i$ located at $x=i-1$. The buyer is assumed to incur a multiplicative loss as they walk along the number line. We let $d\in[0,1]$ be their \emph{multiplicative discount factor}, which measures how much profit is conserved for each step
\footnote{In our model we have a multiplicative decay of the utility which is common in economics literature. We could consider other discount functions, e.g. linear decay $1-p_i-\kappa \cdot i\cdot d$ for some $\kappa>0$, the resulting model is qualitatively similar to the one presented in this work.}
, i.e. the buyer receives $(1-p_i)\,d^{i-1}$ if they choose vendor $i$. For a fixed pricing strategy $\{p_i\}_{i=1}^{n}$, the buyer now chooses an
$
    i^* \in \underset{{i=1\ldots n}}{\text{argmax}} \left\{ (1-p_i)\, d^{i-1} \right\}.
$
In fact, the buyer is indifferent to any convex combination of all suitable $i^*$s. We assume the popsicle vendors incur a loss of an amount $\alpha \in [0,1]$ of the surplus generated by the buyer. Vendor $i^*$ then receives a total of $(1-\alpha)\, p_i$ profit, while every $i\neq i^*$ receives a profit of zero. Finally, we allow also the buyer to make a \emph{side payment} $q\in[0,1]$ to the vendor they chose. This allows the vendor to circumvent the inefficiencies of the sandwich attack and the burden of inserting their own transactions. Thus, a strategy $S$ can be described as a tuple $S = (p_1, p_2, \ldots, p_{n};\,q;\,i^*)$.  For a given choice of  strategy profile $S$, the utility of the buyer satisfies,
\begin{equation}
    u_0(S) = \mathbb{E}\left[(1-p_{i^*} - q)\, d^{i^*-1} \right],
\end{equation}

\noindent where the randomness is taken over the random choice of pure strategies, and the utility of vendor $i$ is,
\begin{equation}
    u_i(S) = \Pr[i=i^*] \cdot \mathbb{E}\left[(1-\alpha)\, p_{i^*} + q\right].
\end{equation}
We interpret $u(S)$ as an element of $\mathbb{R}^{n+1}$ indexed from 0 to $n$. With slight overload of notation, we denote also by $S$ a mixed strategy profile, i.e. $S$ is a probability distribution on the set of all pure strategy profiles. In this case, we denote by $u(S)$ the \emph{expected} utility and thus assume all parties to be risk neutral.

A strategy profile $S$ is said to be an \emph{equilibrium} if no unilateral deviation by any of the parties (i.e. the vendors or the buyer) results in that party receiving a strictly higher expected utility. Note that, in our game $G$, the buyer observes the actions taken by the vendors, and thus makes their choice $q$ and $i^*$ knowing the pricing of all the vendors. An equilibrium is said to be \emph{subgame perfect} (or an SPE) if it is an equilibrium for every subgame of the original game. This means that $S$ is a subgame perfect equilibrium if it also an equilibrium after the prices have been set, i.e. when the buyer has to choose $q$ and $i^*$.
Thus, we may think of $G$ as a multi-leader Stackelberg game where all the vendors are leaders, and the buyer is the follower. Technically speaking $q$ and $i^*$ are functions, $q=q(\cdot)$, and $i^*=i^*(\cdot)$ that assign to each $p=(p_1,p_2,\ldots,p_n) \in \mathbb{R}^n$ a distribution on their respective actions. We say $S=(p;\,q(\cdot);\,i^*(\cdot))$ is an equilibrium if 1) for every $\overline{q}$, and $\overline{i}^*$, 
\begin{equation}
    u_0(S) \geq u_0(p;\,\overline{q}(p);\,\overline{i}^*(p)),
\end{equation}
where and 2) for each $i \in \{1,2,\ldots, n\}$ and every $\overline{p}_i$, with $p' = (p_1,p_2,\ldots p'_i,\ldots,p_n)$, 
\begin{align}
    u_i(S) \geq u_0(p';\, q(p');\, i^*(p')),
\end{align}
Note that in the vanilla setting (without contracts), by backward induction, the buyer will never effectuate the side payment. Hence, in every subgame perfect equilibrium, we have $q=0$. So we may assume for simplicity that the side payments do not exist in the vanilla setting. Thus, we may write a strategy profile as simply $(p_1, p_2, \ldots, p_n;\, i^*)$. We will now characterize the equilibria in the case without commitments.

\begin{theorem}[Vanilla Setting]\label{thm:vanilla}
    Let $S^* = (p_1,p_2,\cdots,p_n;\,i^*)$ be a subgame perfect equilibrium of the popsicle game without contracts, for $0 < \alpha < 1$ and $n\geq 2$. We can characterize $S^*$ depending on the value of $d$ as follows:
    \begin{enumerate}
        \item If $d=1$, then there is perfect competition, and the buyer always receives the popsicle at zero cost, i.e. $u_0 = 1$, and each vendor $j$ receives zero payment, i.e. $u_j=0$.
        \item If $d<1$, then vendor 1 sets their price to $p_1 = 1-d$, vendor 2 sets their price to $p_2=0$, the buyer chooses $i^*\in\{1,2\}$ arbitrarily. The buyer gets utility $u_0=d$, vendor $1$ receives $u_1 \leq \alpha\, (1-d)$, and every vendor $j \geq 2$ receives $u_j=0$.
    \end{enumerate}
\end{theorem}
\begin{proof}
We show the two cases separately.
\begin{enumerate}
    \item For $d=1$, consider the case in which all vendors choose a price of $p_i = 0$ and the buyer chooses a vendor uniformly at random. Then, the utility of the buyer is $u_0 = 1$, which is maximal, and the utility for each vendor $i$ is $u_i = 0$. Suppose that player $i$ deviates to $p_i'>0$. Then, if the buyer includes $i$ in the support of their strategy, this must result in strictly worse utility for the buyer, and they will not include $i$. As such, by deviating to $p_i'>0$, the vendor also receives zero utility because they will not be chosen, which shows the proposed strategy is an equilibrium. Conversely, consider any purported equilibrium where the buyer does not receive utility $u_0=1$. Then it must be the case that $p_{i}>0$ for every $i$. Now let $i'$ be some agent who is chosen by the buyer with probability $\leq \frac12$. Note that such an $i'$ must exist since there are least two vendors. Now consider $i'$ deviating to $p_{i'} = \frac12 \min_{i} p_i$. Observing this, the buyer deviates to selecting $i'$ with probability 1. In this case, $u_{i'} = (1-\alpha)\, p_{i'} / 2$, which, by assumption, is strictly better. This shows that the strategy considered is not an equilibrium.
    
    \item For $d<1$, the buyer strictly prefers the earlier vendors over the later ones for the same price. Once again, consider the scenario where vendor 1 sets their price at $p_1=1-d$, vendor 2 sets their price at $p_2=0$, and the remaining vendors set their prices arbitrarily. The utility of the buyer for choosing vendor 1 is $u_0 = (1-(1-d))\, d^0 = d$, while the utility for choosing vendor 2 is $u_0 = (1-0)\,3 d^1 = d$. For any other vendor $i \geq 3$, their utility to the buyer is $u_0 \leq d^2 < d$ since $d<1$. As a result, the buyer is happy to choose any convex combination of vendor 1 and vendor 2, while the inclusion of any vendor $\geq 3$ would lead to strictly less utility. Similarly, if vendor 1 increases their price, the buyer will deviate to vendor 2, while if vendor 1 decreases their price, their profit will diminish. This holds analogously for vendor 2, although they are unable to lower their price any further. This shows the proposed strategy is sufficient for an equilibrium. Necessity follows from a simple calculation analogous to the first part of the theorem. \qed
\end{enumerate}
\end{proof}

\section{A Stackelberg Attack}
In this section, we demonstrate the existence of a Stackelberg attack on the popsicle game described above. The attack allows the first vendor to extract the full value from the buyer and sell the item at full price by carefully committing to a strategy that itself depends on the commitments made by the other parties.

To proceed, we first formally define what it means for a player to commit to a strategy. We will use the model of \cite{stack_attack_mathias}, which we will now detail for self-containment. First, we fix an extensive-form representation of a game $G$. We then consider a \emph{commitment move}, which is a node with fanout 1 and is labelled with the index of a player. We denote by $C^i(G)$ the game that results from adding a commitment move for player $i$ as the new root of the game, of fanout 1, and letting its subgame be $G$. Intuitively, a contract move signifies that the corresponding player is allowed to commit to a strategy in the subgame. Formally, a commitment move is `syntactic sugar' for the larger expanded tree that encodes all possible commitments that this player can make. To model this, we compute the set of all cuts that the corresponding player can make in that subgame. Here, a `cut' for player $i$ is any subtree where some of the outgoing edges from nodes owned by $i$ are removed such that 1) there is always a path from every subgame to a leaf (labeled with a utility vector), and 2) the cut respects information sets, in the sense that the agent is not allowed to cut away a set of moves from only one node in an information set. An illustration of this procedure is shown in \cref{fig:expanded_tree}.

\begin{figure}
    \centering
    \begin{tikzpicture}[level 1/.style={sibling distance=4em}]
        \node[rectangle,draw] {$2$}
        child {node[circle,draw] {$1$}
        child {node[circle,draw] {$2$}
          child {node {$\bullet$}}
          child {node {$\bullet$}}
        }
        child {node {$\bullet$}}};
        \node at (2,-2.3) {=};
    \end{tikzpicture}
    \begin{tikzpicture}[level 1/.style={sibling distance=8em},level 2/.style={sibling distance=4em}]
        \node[circle,draw] {$2$}
        child {node[circle,draw] {$1$}
            child {node[circle,draw] {$2$}
              child {node {$\bullet$}}
              child {node {$\bullet$}}
            }
            child {node {$\bullet$}}}
        child {node[circle,draw] {$1$}
            child {node[circle,draw] {$2$}
              child {node {$\bullet$}}
              child[draw=white] {node {~}}
            }
            child {node {$\bullet$}}}
        child {node[circle,draw] {$1$}
            child {node[circle,draw] {$2$}
              child[draw=white] {node {~}}
              child {node {$\bullet$}}
            }
            child {node {$\bullet$}}};
    \end{tikzpicture}
    \caption{Computing the expanded tree $C^2(G)$ for a simple game $G$. The contract move for player 2 is depicted as a square node. To expand, we compute all cuts in the game for player 2 and attach them to a node belonging to player 2.}
    \label{fig:expanded_tree}
\end{figure}
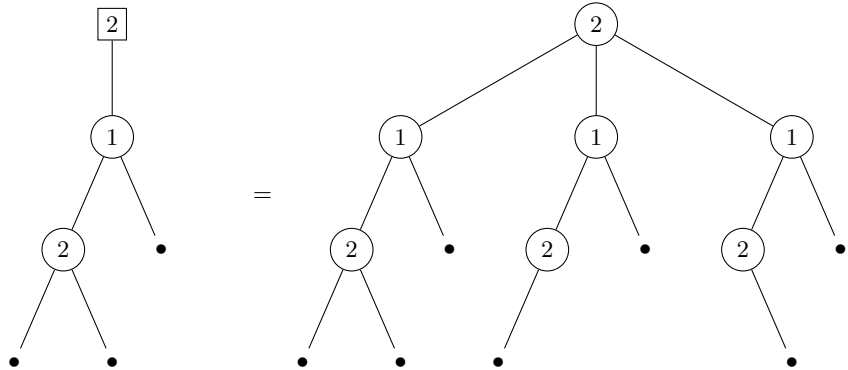

As shown in \cite{stack_attack_mathias}, the game $C^i(G)$ is the Stackelberg version of the game $G$ with the player $i$ being the leader. To account for multiple commitments, we may simple repeat this process, expanding the commitments moves in a bottom-up manner. For instance, in the game $C^1(C^2(G))$, we first expand the game with all cuts for player 2, and subsequently expand the game for player 1. Note that player 1 can cut, conditioned on the commitment made by player 2. This complex conditioning is what makes these attacks especially potent. In the literature, this is known as a \emph{reverse Stackelberg equilibrium} for the case of two commitments, though this model allows us to repeat this process \emph{ad libitum}. 
Our goal now is to show that the popsicle game is \emph{not} Stackelberg resilient. In particular, we will show that accounting for absolute commitments changes the incentives of the game and results in a different equilibrium. To formalize this, we use the definition of \cite{stack_attack_aamas}, which we state here for completeness.
\begin{definition}[Stackelberg Resilience, \cite{stack_attack_aamas}]
    Let $G$ be an $n$-player game in extensive-form
    , and let $C^i(\cdot)$ be the operator that computes the expanded tree that results from adding a commitment move to player $i$
    . We say that $G$ is \emph{Stackelberg resilient} if for  for every permutation $\pi : [n] \rightarrow [n]$, and every equilibrium $\overline{s}^*$ in $C^{\pi(1)}(C^{\pi(2)}(\cdots(C^{\pi(n)}(G)))$, there is an equilibrium $s^*$ in $G$, such that for every $i \in [n]$, it holds that,
    $
        u_i(s^*) = u_i(\overline{s}^*).
    $
\end{definition}

\noindent We are now in a position to state the main result of this work.
\begin{theorem}[Main Result]\label{thm:main}
    For a buyer with discount factor $d < 1$, $n\geq 2$ vendors, and any tax rate $0 \leq \alpha < 1$, the popsicle game is not Stackelberg resilient.
\end{theorem}
\begin{proof}
    To show this statement, we proceed by construction. We have to show the existence of an equilibrium with contracts that does not exist in the vanilla case. Suppose the agents have contracts in the order $1 \rightarrow 2 \rightarrow \cdots \rightarrow n \rightarrow 0$ (recall that 0 is the buyer), and consider the expanded tree as defined above. For a subgame, after the contracts have been deployed, we say a vendor $j$ has \emph{committed} to a price of $p_j^*$ if letting $p_j=p_j^*$ is an SPE in that subgame. Now consider the following contract for vendor 1.
    \begin{quote}
        \textbf{if} there is a vendor $j$ that commits to a price of $p_j \neq 1$ \textbf{then}
        \begin{quote}
            commit to a price of $p_1 = 0$
        \end{quote}
        \textbf{else if} the buyer commits to $i^*=1$ and a side payment of $q=1$ \textbf{then}
        \begin{quote}
            commit to a price of $p_1=0$
        \end{quote}
        \textbf{else}
        \begin{quote}
            commit to a price of $p_1=1$
        \end{quote}
        \textbf{end}
    \end{quote}

   We claim that it is an equilibrium for vendor 1 to deploy this contract, and for all the other agents to comply with the contract, i.e. each other vendor $j$ lets $p_j = 0$ in every subgame, and the buyer chooses $i^*=1$ and $q=1$ every time. Here, vendor 1 gets $u_1=1$ utility, while every other agent $j\neq 1$ receives $u_j=0$ utility. From the perspective of vendor 1, the obtained utility of $u_1=1$ is optimal, so no deviation can be strictly better. For vendor $j \geq 2$, if they comply and commit to a high price of $p_j = 1$, the buyer will receive zero utility by deviating to $i^*=j$. If instead we have $p_j<1$, vendor 1 will automatically set their price at $p_1=0$.  In this case, by backwards induction the buyer must choose $i^*=1$, and so vendor $j$ receives zero utility regardless of their strategy. For the buyer, if they commit to $i^*=1$ and a side payment of $q=1$, they will obtain zero utility. On the other hand, if they deviate to any $i^* = j> 1$, they will also receive zero utility, since each such vendor has $p_j=1$. Thus, no deviation for the buyer can increase their utility. This shows that vendor 1 deploying the contract is indeed an equilibrium.  Furthermore, by \cref{thm:vanilla}, the purported strategy is not an equilibrium in the vanilla case. \qed
\end{proof}

We have shown the existence of a different equilibrium. Note that the first vendor is well positioned to sweeten the deal if we want an even more insidious attack.  The buyer's commitment can be made more attractive if the vendor promises a price of $q=1-\varepsilon$, where $\varepsilon$ is an arbitrary amount of currency.  Similarly, vendor 1 could spread some of their earnings down the line.  For example, they could commit to sending $\varepsilon$ to later vendors if they comply. This would clearly break down if there were too many vendors that needed buying off.  Given that smart contracts have the capacity to call randomness using a randomness beacon (in a computational sense), one of the committed vendors could be chosen at random to receive the bribe, thereby raising the expected utility of all strictly above the utility $0$ of not committing. 

\bibliographystyle{splncs04}
\bibliography{refs.bib}
\end{document}